\documentclass[12pt]{article}
\usepackage{pifont,xspace,epsfig, wrapfig}
\usepackage{graphicx}
\usepackage{latexsym}
\usepackage{amssymb}
\usepackage{graphics}
\usepackage{color}
\usepackage[usenames,dvipsnames,svgnames,table]{xcolor}
\usepackage{amsthm}
\usepackage{amsmath}
\usepackage{amssymb}
\usepackage{float}
\usepackage{algpseudocode}

\usepackage{subfig}
\usepackage{tikz}
\usepgflibrary{shapes.symbols} % LATEX and plain TEX and pure pgf
\usepgflibrary[shapes.symbols] % ConTEXt and pure pgf
\usetikzlibrary{shapes.symbols} % LATEX and plain TEX when using Tik Z
\usetikzlibrary[shapes.symbols] % ConTEXt when using Tik Z
\usepgflibrary{shapes.geometric} % LATEX and plain TEX and pure pgf
\usepgflibrary[shapes.geometric] % ConTEXt and pure pgf
\usetikzlibrary{shapes.geometric} % LATEX and plain TEX when using Tik Z
\usetikzlibrary[shapes.geometric] % ConTEXt when using Tik Z

\newcommand{\dahntab}[1]{
  \newbox\mybok%
  \setbox\mybok=\hbox{\vbox{
      \begin{tabbing}
        #1
      \end{tabbing}%
    }}

  \newdimen\bokwidth%
  \bokwidth=\wd\mybok%
  \newdimen\myl%
  \myl=\textwidth%
  \divide\myl by 2%
  \divide\bokwidth by -2%
  \advance\myl by\bokwidth%
  \vrule width\myl height 0pt depth 0pt%
  \usebox\mybok%
}

\floatstyle{ruled}
\newtheorem{definition}{Definition}
\newtheorem{lemma}{Lemma}
\newtheorem{example}{Example}
\newtheorem{proposition}{Proposition}
\newtheorem{theorem}{Theorem}
\newtheorem{corollary}{Corollary}
\newtheorem{observation}{Observation}

\begin{document}
% Title must be 150 characters or less

\title{A Parametric Worst-Case Approach to Fairness in TU-Cooperative Games.}
%\author{Cosmin Bonchi\c{s}, Gabriel Istrate\footnote{Corresponding Author}}
\author{Cosmin Bonchi\c{s} \thanks{Department of Computer Science, West University of Timi\c{s}oara,
Bd. V. P\^{a}rvan 4, Timi\c{s}oara, RO-300223, Romania and e-Austria Research Institute, Bd. V. P\^{a}rvan 4, cam. 045
B, Timi\c{s}oara, RO-300223, Romania},  Gabriel Istrate\thanks{Corresponding Author. Center for the Study of Complexity, Babe\c{s}-Bolyai University,
F\^{a}nt\^{a}nele 30, cam. A-14, Cluj Napoca, RO-400294, Romania and
 e-Austria Research Institute, Bd. V. P\^{a}rvan 4, cam. 045
B, Timi\c{s}oara, RO-300223, Romania. email: gabrielistrate@acm.org}}

%\titlerunning{Fairness in TU-Cooperative Games}
%\authorrunning{Bonchi\c{s}, Istrate}

\maketitle

\begin{abstract}

We propose a parametric family of measures of fairness in allocations of TU-cooperative games. Their definition is based 
on generalized R\'enyi Entropy, is related to the  Cowell-Kuga generalized entropy indices in welfare economics, 
and aims to parallel the spirit of the notion of price of anarchy in the case of convex TU-cooperative games.

Since computing these indices is NP-complete in general, we first upper bound the performance of a ``reverse greedy'' algorithm 
for approximately computing worst-case fairness. The result provides a general additive error guarantee in terms of two (problem dependent) packing constants. 
We then particularize this result to the class of induced subset games.
For such games computing worst-case fairness is NP-complete, and the additive guarantee constant can be explicitly computed. 
We compare this result to the performance of an alternate algorithm based on ``biased orientations''. 
\end{abstract}

\section{Introduction}

A central issue in cooperative game theory is coming up with a "fair" allocation of a total cost to a set of
agents that respects coalitional rationality. The 
Shapley value \cite{roth1988shapley} provides the classical, well-understood approach to this problem. But it
is not the only possible solution, nor is it appropriate from all points of view. Indeed, computing the 
Shapley value  is often intractable \cite{deng1994complexity}. Imposing the Shapley value in a distributed 
multiagent setting (though possible in principle) may be even more problematic, as it requires the use of a (centralized) 
mechanism that incentivizes the players to commit to accepting this cost-sharing method. Furthermore, there 
are settings where using the Shapley value may simply be inappropriate. This is, for instance, the case of 
coalitions with a dynamically evolving structure, or that of games in which the Shapley value is not in 
the core (hence it is not immune against coalitional deviations). Finally, the Shapley value may be 
incompatible with the presence of prior social constraints that favor alternative social arrangements 
(e.g. egalitarianism \cite{dutta1989concept}).

We do not mean by the previous discussion that one could replace the Shapley value by any alternative allocation in the core: it simply may 
be the case that no single solution concept is appropriate in all circumstances. 

In noncooperative game theory the seminal work of Roughgarden and Tardos \cite{roughgarden2002bad,selfish-routing,inefficiency-equilibria} 
on the price of anarchy provides a powerful intuition that could have an analog 
in cooperative settings as well. Instead of advocating any particular solution to the problem of equilibrium 
selection, their price of anarchy measure takes a pessimistic perspective, quantifying the degradation in 
overall performance due to uncoordinated behavior, measured on the {\em worst} equilibrium. This circumvents 
the problem of equilibrium selection by providing (pessimistic) guarantees that are valid for {\em any} rational solution. 

In this paper we propose an approach with a similar philosophy for cooperative games. Rather than attempting 
to postulate any particular solution of such a game, we  investigate the fairness of an arbitrary ``rational'' 
cost allocation (where in this paper we define rationality as membership in the core). We propose not a 
single measure, but a parametric family of measures of fairness, based on variations on the concept of 
{\em R\'{e}nyi entropy} and fruitfully used before as measures of inequality in welfare economics 
\cite{CowellKuga1981}. 

The structure of the paper is as follows: in Section 2 we briefly review some relevant concepts in cooperative game theory, 
information theory and combinatorial optimization. We then introduce in Section 3 our parametric family of measures of fairness. 
In Section 4 we first present a ``reverse greedy'' algorithm  for the problem of 
minimizing the worst-case fairness of a convex cooperative game. We give a general result 
on the performance of this algorithm. In section 5 we particularize our discussion to the class of {\em induced subgraph games} of 
Deng and Papadimitriou \cite{deng1994complexity}: first we apply the general result in the previous section to obtain an explicit bound 
on the performance of the reverse greedy algorithm. Then we display a connection with a weighted version of the minimum entropy 
orientation problem from \cite{cardinal2008minimum} that allows us to obtain a bound that is better for some particular values of the parameter $\lambda$. 
 
\section{Preliminaries}\label{preliminaries}

We will work in the framework of Cooperative Game Theory 
(for a recent survey of the relevant literature from a Computer Science perspective see \cite{chalkiadakis2011computational}). 
We will assume knowledge of basic concepts from this literature such as cooperative game, core, concavity/convexity of a cooperative game, 
imputation, Shapley value and so on. 

We need the following notion adapted from \cite{arin2008axiomatic}: 
\begin{definition} 
A {\em solution concept $q$} is a function that assigns to every convex cooperative game $\Gamma$ an imputation $q=q(\Gamma)$. 
It is a {\em core concept} if $q(\Gamma)\in core(\Gamma)$ for every such game $\Gamma$.
\end{definition} 

We will also assume knowledge of computational complexity. A standard recent reference is \cite{arora-barak-book}. 

We recall the notion of R\'{e}nyi entropy of a probability distribution: 

\begin{definition} \cite{cover1991elements}
For any discrete random variable $X$ with probability
mass function $p=(p_{i})_i$ and for any real $\lambda>0,$ $\lambda\neq 1$
the {\em R\'{e}nyi entropy of order $\lambda$ of $X$} is defined as:

\begin{align}
H_{\lambda}(X) & =\frac{1}{1-\lambda}\log_{2}\left(\sum_{i}p_{i}^{\lambda}\right).\label{eq:RenyiEnt}
\end{align}
\end{definition}

We can complete this definition for $\lambda = 1$ by the usual Shannon entropy 
\begin{align}
H(X)=H_{1}(X) & =-\sum_{i} p_{i}\log_2 p_{i}.\label{eq:ShannonEnt}
\end{align}

Without loss of generality we will assume  in the sequel that one can convert any vector of nonnegative values $X$ into a probability distribution 
by normalization. Without risk of confusion we will use the name $X$ for the resulting distribution as well. 

We need the following simple fact, which informally states that stochastic domination leads to lower R\'{e}nyi entropy:
\begin{lemma}
Let $P=(p_1, \dots, p_i, \dots, p_j, \dots, p_n)$ be a probability distribution with $p_1\geq p_2 \geq \dots \geq p_n$. 
Let $1 \leq i < j \leq n$ and $0 < \epsilon \leq p_j.$
Defining $Q=(p_1, \dots, p_i + \epsilon, \dots, p_j - \epsilon, \dots, p_n)$ for any $\lambda > 0$ we have:  
\begin{align}
 H_{\lambda}(P)> H_{\lambda}(Q).
\end{align}
\label{dominate_min_ent}
\end{lemma}
When $\lambda = 1$ the result was proved in \cite{AlonOrlitsky_sourcecoding} (see also \cite{cardinal2008minimum}). 
The general case follows from investigating (using standard calculus techniques) the monotonicity of function $f(x)=(x+\epsilon)^{\lambda}-x^{\lambda}$ for $x>0.$

We will highlight the concepts we introduce in this paper on the class of  \emph{induced subgraph
games}, studied in \cite{deng1994complexity} (see also \cite{chalkiadakis2011computational,topkis1998supermodularity}).

\begin{definition} 
{[}Induced subgraph (IS) games:{]} We are given a connected loopless graph $G=(V,E)$ and
a set of integer \emph{weights} $(w_{i,j})_{(i,j)\in E}$ on the edges.
Vertices of $G$ are interpreted as \emph{players} in a cooperative
game, that may contribute to a coalition the edges they have in common with members of
their coalition. In other words, given set $S\subseteq V$,
the \emph{value of coalition} $S$ is 
\begin{equation} 
v(S)=\sum_{(i,j)\in E,i,j\in S}w_{i,j}.
\label{val-is}
\end{equation} 
\end{definition} 

In \cite{deng1994complexity} it is proved that for nonnegative weights
induced subgraph games are convex. In what follows we will assume that weights are nonnegative. In addition 
we also assume that, for any vertex $v \in V,$ the sum of weights of its adjacent edges is strictly positive.   
Otherwise we could simply eliminate $v$ from $G$ without changing the problem.

\begin{definition}
\label{cover-mis}
Given an instance $(G,w)$ of IS, a \emph{cover of $(G,w)$} is a function $u:V\rightarrow{\bf \mathbb{Z}}_{+}$
such that $u(V)=v(V)$ and for all $S\subseteq V$, 
\[
u(S):\stackrel{def}{=}\sum_{i\in S}u(i)\leq v(S).
\]
\end{definition}

That is, a cover is an imputation in the core of the cooperative game $(V,v)$ having integer values. The restriction 
to integers is computationally motivated, and reasonable in any intended \textit{algorithmic} application of cooperative game theory.

To apply our result to IS games we need a problem generalizing the 
{\em minimum entropy orientation problem} from  \cite{cardinal2008minimum}:

\begin{definition}
\label{bias-orientation}Consider a connected loopless graph $G=(V,E)$
and a set of integer \emph{weights} $(w_{i,j})_{(i,j)\in E}$. 
\begin{itemize} 
\item An {\em orientation} of $G$ is a mapping $h:E\rightarrow V$ such that for any $e=(i, j)\in E$, $h(e)$ is a vertex of $e$. If $h(e)=i$ we will say that $e$ is oriented towards $i$ in $h$. 

\item an orientation $h$ of $G$ is called \emph{biased} if each edge $(i,j)\in E$
with $v(i)>v(j)$ is oriented toward $i$, where the value $v(i):=\sum_{e\sim i}w_{e}$
is the weight sum of all edges adjacent to vertex $i$. In other words, edges
are oriented towards the vertex whose weight sum is larger. 
\end{itemize} 
\end{definition}

An optimization problem associated to edge orientations is to find an orientation $h$ of $G$ that minimizes 
R\'{e}nyi entropy of the weights vector:

\[
 H_{\lambda}[h]=\frac{1}{1-\lambda}\log\left[\sum_{i\in V}\left(\frac{\sum_{e\in h^{-1}(i)}w_{e}}{\sum_{e\in E}w_{e}}\right)^\lambda\right]
\]

(and similarly for $\lambda=1$). We will denote this problem by MREWO.

One can find an approximately solution to any instance of this problem by means of the Greedy algorithm 
(formally discussed in the unweighted case in \cite{cardinal2008minimum,cardinal2008pmean}), described informally as follows:
\begin{itemize}
 \item Select a vertex $i \in V$ that maximizes value $v(i)$.
 \item Orient all edges in $G$ towards $i$.
 \item Delete $i$ and all its adjacent edges from $G$ and proceed recursively. 
\end{itemize}

\section{The (parametric) worst-case fairness of a cooperative game}

We now define the main object of interest, a parametric family of 
measures of fairness for cost allocations of a cooperative game $\Gamma=(N,v)$. They are indexed by: 
\begin{itemize} 
\item A positive real $\lambda$. 
\item A solution concept $q$, yielding vector $q(\Gamma)\in {\bf R}_{+}^{|N|}$. Intuitively vector $q(\Gamma)$ represents a baseline "standard of fairness" 
to which all other elements are held. Our measures attempt to evaluate the largest possible discrepancy between 
an imputation in the core and $q(\Gamma)$. 
\end{itemize} 

To define our measure we need the following concepts: 

\begin{definition}
Let $P=(p_i)$ and $Q=(q_i)$ be two distributions and $\lambda > 0.$ 
\begin{enumerate}
 \item The R\'enyi divergence of order $\lambda$ is defined by:

\begin{eqnarray*}
D_{\lambda}(P\parallel Q) & = & \frac{1}{\lambda-1}\log\left(\sum_{i}p_{i}^{\lambda}q_{i}^{1-\lambda}\right)
\end{eqnarray*}

\item The discrete R\'{e}nyi relative entropy of order $\lambda$ is:
\begin{eqnarray}
h_{\lambda}[P,Q] & = & \frac{1}{1-\lambda}\log\left(\sum_{i}q_{i}^{\lambda-1}p_{i}\right)+\frac{1}{\lambda}\log\left(\sum_{i}q_{i}^{\lambda}\right)\notag\\
& & -\frac{1}{\lambda(1-\lambda)}\log\left(\sum_{i}p_{i}^{\lambda}\right)
\label{eq:relativeRenyi}
\end{eqnarray}
\end{enumerate}
\end{definition} 

The following is the discrete version of the Gibbs inequality (Lemma 1 from \cite{lutwak2005crame}):

\begin{lemma}
For any $\lambda>0$ the discrete $\lambda$-R\'{e}nyi relative entropy has a nonnegative
value.\\
\begin{eqnarray*}
h_{\lambda}[P,Q] & \geq & 0.
\end{eqnarray*}
\label{positiveRenyiRelative}
\end{lemma}
For the proof see the Appendix.

Though they do not generally yield metrics, informational divergence measures have a significant history of use as indicators of "similarity" between
 two probability distributions. They are pseudometrics and have a number of other desirable properties that have made them applicable to problems such 
as learning and classification  (to give just one example we refer to \cite{ullah1996entropy}). 
This observation enables us to finally give the following definition of worst-case fairness: 
\begin{definition} 
Given cooperative game $\Gamma=(N,v)$ and real number $\lambda > 0$ the {\em $\lambda$-worst-case fairness of game $\Gamma$} with respect to solution concept $q$ is defined as 
\[
Fair_{\lambda}(\Gamma,q)=sup\{ D_{\lambda}(x||q(\Gamma)):x\in core(\Gamma)\cap {\bf Z}^{N} \}. 
\] 
\label{main-def}
\end{definition}

Definition~\ref{main-def} obviously depends on the choice of the baseline solution concept $q$. At least several special cases make sense: 

\begin{enumerate} 
\item{\em strictly egalitarian worst-case fairness: } $q(\Gamma)$ is the uniform vector $U(i)=\frac{v(N)}{|N|}$ for all $i \in N.$ 
That is, we benchmark all solutions against the uniform distribution.

 Though it could seem at first somewhat controversial from a modeling standpoint, as it requires a very strong form of equality, the study of this 
measure makes sense at least from a {\em mechanism design} point of view. Namely, in the case of convex/concave games, particularly interesting examples
 of imputations in the core arise from group-strategyproof mechanisms or, equivalently, cross-monotonic sharing schemes (see \cite{moulin1999incremental} 
and Chapter 15 of \cite{nrtvCUP07}). Requiring cross-monotonicity yields (according to \cite{immorlica2008limitations}) a ``plausible
 notion of equity''.  A natural issue raised by this scenario is the potential inequity in cost assignments induced by the use of cross-monotonic mechanisms. 
The strictly egalitarian worst-case fairness offers a pessimistic measure of this inequity. 

The study of strictly egalitarian worst-case fairness can be further justified on technical grounds: 
the computational tractability of the uniform distribution makes the derivation of algorithmic bounds on worst-case fairness easier than the general case.
 Furthermore, as we will see, we can derive weaker bounds for some alternate core concepts
 from bounds on strictly egalitarian worst-case fairness. 

\item{\em marginalist worst-case fairness}: in this case $q$ is the probability distribution obtained from the Shapley value by normalization. 
One should, of course, not expect this measure to be efficiently computable given that the Shapley index is hard to compute. 

\item{\em egalitarian worst-case fairness}: in this case $q$ corresponds to the egalitarian solution of Dutta and Ray \cite{dutta1989concept}. 
Alternately one could use the equivalent Arin and I\~{n}arra solution concept \cite{arin2001egalitarian} that is a core concept for convex games.  
We will not study this measure in the present paper.
\end{enumerate}

\begin{example} 
{\em $ $ Consider the induced subgraph game $\Gamma$ with three players presented in Figure~\ref{IS-three}(a). The total 
 payoff to be shared between players is $12=2+4+6$.  The core of $\Gamma$ is given by the system of equations in Figure~\ref{IS-three}(b). 
 
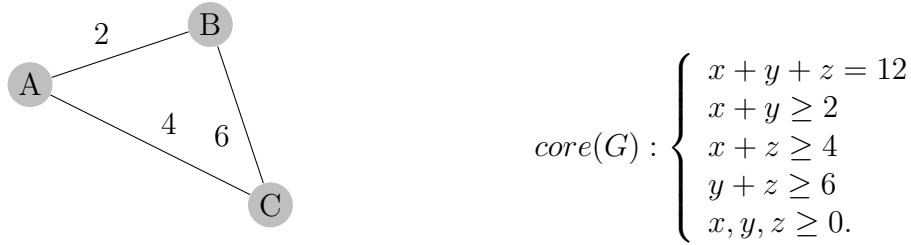
\begin{figure}
 
\begin{minipage}[t]{0.50\textwidth}
\begin{center}
\begin{tikzpicture}
[elips/.style={shape=ellipse, inner ysep=1.8cm, inner xsep=2.5cm},
link/.style={-, draw=gray!50},
vertex/.style={circle,fill=black!25,minimum size=17pt,inner sep=0pt}
]
\node[name=E,elips]{};

\matrix [
  column sep={0.8cm,between origins},
  row sep={0.8cm,between origins}, 
  ] at (E.mid)
{
 & & &  \node[vertex](n2) {B}; & & \\
 \node[vertex] (n1) {A}; & & & & &  \\
 & & & & & \\
 & & & & \node[vertex] (n3) {C}; &  \\
};

\draw [link, draw=black!100] (n1) to node[auto]{2} (n2);
\draw [link, draw=black!100] (n1) to node[auto]{4} (n3);
\draw [link, draw=black!100] (n3) to node[auto]{6} (n2);
\end{tikzpicture}
\end{center}
\end{minipage}
%\vspace{-1cm}
%\end{center}
\hfill
\begin{minipage}[t]{0.35\textwidth}
\begin{center}
\vspace{1.2cm}
 \begin{align}
core(G):  \left\{\begin{array}{l}
x+y+z=12 \notag \\
x+y\geq 2 \notag\\
x+z\geq 4 \notag \\
y+z\geq 6 \notag\\
x,y,z\geq 0.\notag 
\end{array}
\right.
 \end{align}
\end{center}
\end{minipage}
\vspace{-0.4cm}
\caption{(a) An IS game with three players; (b) Its core.}
\label{IS-three}
\end{figure}

The Shapley value is (according to the computation in \cite{deng1994complexity}) $Sh=(3;4;5)$. 
Furthermore, by performing simple symbolic computations one can show the following: 
\begin{itemize}
\item There are a total of 57 integer imputations (covers) in the core. Of these six are extremal (i.e. corners of the polygon geometrically 
describing the core). 
\item There are two covers  in the core maximizing strictly
egalitarian worst-case fairness (for $\lambda =1$). They are $(2;0;10)$ and $(0;2;10)$ respectively. Worst-case fairness $Fair_1(\Gamma, U)$ is approximately 0.934.
\item On the other hand there exists an unique optimal cover in the case of marginalist worst-case fairness : $X=(4;8;0)$. The corresponding worst-case fairness 
$Fair_{1}(\Gamma, Sh)$ is approximately 0.805. 
\end{itemize} 

Thus employing the Shapley value as our standard of fairness of general cost allocations in the core is, at least in this game, more fair than imposing perfect equality.
}
\label{example-1}
\end{example}

\subsection{Further comments} 

We have limited our choices to {\em imputations with integer values}. Though certainly non-standard (and potentially controversial from a ``pure'' game-theoretic 
perspective), we feel that such a restriction is warranted in any {\em algorithmic} result on computational game theory. Indeed, ``real-life'' monetary payments are
 not continuous but discrete. We believe that it is not unnatural to impose such a restriction on imputations in a paper dealing with algorithmic aspects. 

The problem of computing the strictly egalitarian worst-case fairness has an especially tractable interpretation. 
In this case maximizing R\'{e}nyi divergence is equivalent to the problem of minimizing R\'{e}nyi entropy. A more 
limited connection between divergence and entropy holds even in the general case. Given distribution $R=(r_i),$ denote $r_{max} = max\{r_j:j\in supp(R)\},$
 $r_{min} = min\{r_j:{j\in supp(R)}\}$ and define
 \begin{equation} 
 nu(R)=\log\left(\frac{r_{max}}{r_{min}}\right),
 \end{equation} 
 {\em the nonuniformity} of distribution $R.$ 

\begin{lemma}
 Let $P,Q,R$ be probability distributions and $\lambda >0$. We have:
 \begin{equation*} 
 H_{\lambda}(Q)-H_{\lambda}(P)- nu(R)\leq D_{\lambda}(P||R)-D_{\lambda}(Q||R)\leq H_{\lambda}(Q)-H_{\lambda}(P)+nu(R)
 \end{equation*}  
\label{PQR}
\end{lemma} 
\vspace{-0.5cm}
\begin{proof}
See Appendix.  \qed
\end{proof}

Computing a cover $x$ of a cooperative game $\Gamma=(N,v)$ that witnesses the value of the 
$\lambda$-worst-case fairness $Fair_{\lambda}(\Gamma,q)$ has, when $q$ is the uniform distribution an alternative combinatorial description 
relying on concepts of submodular 
optimization \cite{grotschel1993geometric}. In this setting a convex cooperative game $\Gamma=(N,v)$ 
maps to an instance of the {\em supermodular} Set Cover problem. The more frequent problem, that
 of {\em submodular} set cover \cite{wolsey-submodular,fujita-ssc} corresponds to the {\em dual} 
of the cooperative game $\Gamma$, that is to the game $\Gamma^{*}=(N,v^{*})$, with $v^{*}(S)=v(N)-v(N\setminus S)$.
 As games $\Gamma$ and $\Gamma^{*}$ have the same core, the measure $Fair_{\lambda}$ can be computed with respect to the convex game $\Gamma^*$ instead of $\Gamma$.

Finding a cover realizing optimal value $Fair_{\lambda}(\Gamma, U)$ 
is a natural generalization of a extension (studied in  
\cite{cardinal2008pmean}) of the minimum entropy set cover problem \cite{halperin-karp-sc,tight-minentropy-setcover}.
This problem specializes to induced subgraph games as follows: 

\begin{definition} 
\textbf{{[}Minimum R\'{e}nyi Entropy IS{]} (MREIS):}\\
{[}GIVEN:{]} An instance $(G,w)$ of IS, and a real number $\lambda\in\mathbb{R}_{+}, \lambda\neq1 .$  \\ 
{[TO FIND:]} A cover $u$ of $(G,w)$ that minimizes : 
\[
H_{\lambda}(u)=\frac{1}{1-\lambda}\log\left[\sum_{i\in V}\left(\frac{u(i)}{u(V)}\right)^{\lambda}\right].
\]

\label{def-mpmis}
\end{definition}

Replacing $H_{\lambda}$ by Shannon entropy $H_1$ in the above definition we get the {\em minimum entropy 
induced subset} problem. 

Another version of MREIS involves \textbf{fractional covers.}
We no longer insist on the integrality condition in Definition~\ref{cover-mis}.
 Instead we allow imputations to be real-valued. Clearly, this results in a larger solution set.

\section{Results}

Many results related to submodular optimization analyze the performance
of the GREEDY algorithm. For \emph{supermodular} (convex) games, this
doesn't quite make sense: assigning the first element $i$ its payoff
$v(\{i\})$ does not take into account the fact that the contribution
of player $i$ increases with the coalition, being largest for the
coalition $N\setminus\{i\}$. In other words $v(\{i\})\leq v(N)-v(N\setminus\{i\})$, 
 and, to create an imbalanced allocation we should assign player $i$ its utopia payoff
(that is the right-hand quantity, rather than the left-hand). This leads
to considering the ReverseGreedy algorithm displayed below. 

\begin{figure}[!h]
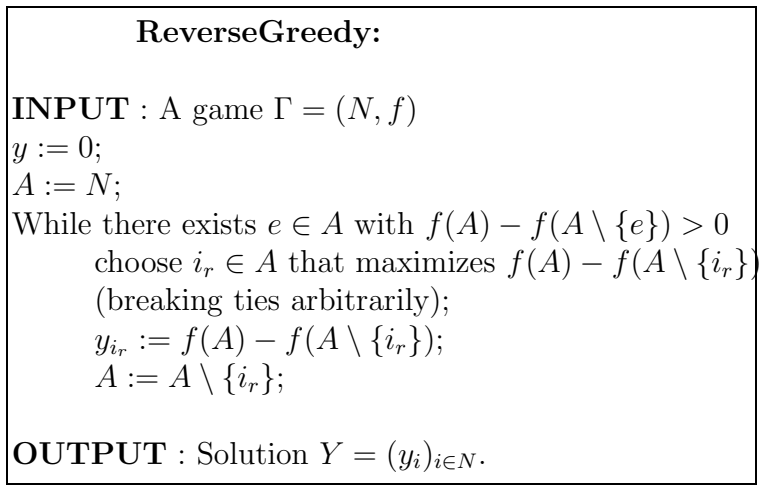

\begin{center}
\fbox{
\begin{minipage}[c]{0.7\columnwidth}
%\begin{algorithm}%{ReverseGreedy}
%\begin{algorithmic}
\begin{center}
\dahntab{
\=\ \ \ \ \=\ \ \ \ \=\ \ \ \ \= 
{\bf ReverseGreedy: }
\\
\\
{\bf INPUT} : A game $\Gamma=(N,f)$ \\
 $y:=0$; \\
 $A:=N$; \\
While {there exists $e\in A$ with $f(A)-f(A\setminus\{e\})>0$} \\
\> \> \> choose $i_{r}\in A$ that maximizes  $f(A)-f(A\setminus\{i_r\})$\\
\> \> \> (breaking ties arbitrarily); \\ 
\> \> \> $y_{i_r}:=f(A)-f(A\setminus\{i_r\})$; \\
\> \> \> $A:=A\setminus\{i_r\}$; \\
\\
{\bf OUTPUT} : Solution $Y=(y_{i})_{i\in N}$.
}
%\end{algorithmic}
%\end{algorithm}
\end{center}
\end{minipage}}
\end{center}
\caption{\label{alg:The-Reverse-Greedy-1}The ReverseGreedy algorithm}
\end{figure}

\begin{example} $ $ {\em Consider the setting of Example~\ref{example-1}. The ReverseGreedy algorithm computes one of the covers $(0;2;10)$ or $(2;0;10)$
(optimal in the case of egalitarian worst-case fairness). 
The computed cover depends on the tie-breaking rule between the first two nodes. 
Indeed, the algorithm first selects node $C$, allocating its utopia value $4+6=10$. Then it selects one of $A$ and $B$ in an arbitrary order.
}
\label{example-2} 
\end{example}

The ReverseGreedy algorithm has an especially attractive interpretation
via \emph{game duality}: ReverseGreeedy on game $\Gamma$ simply corresponds by duality to 
the Greedy algorithm applied on the dual game $\Gamma^{*}.$

\subsection{Main result} 

In this section we give our main result, a bound on the performance of the ReverseGreedy algorithm in approximating the strictly egalitarian worst-case fairness. That is we consider the case when $q=U$, the uniform distribution. 

To do so we have to introduce some notation. 
Specifically: 
\begin{enumerate} 
\item We will denote by $l$ the number of iterations of the ReverseGreedy algorithm. 
\item For $1\leq r\leq l$ denote by $i_{r}$ the element chosen at stage $r$ of the algorithm. 
Let $W_{r}=\{i_{1},\ldots, i_{r}\}$, 
$A_{r}=U\setminus W_{r}$ and $\Delta_{r}$ be the value of element $y_{i_r}$ set at stage $r$. 
\end{enumerate}

We now introduce a quantity, the "impact of $j$ on $i_{r}$", that will play a fundamental role in our results below: 

\begin{definition}
For any $1 \leq r \leq l$ we define the impact of $j$ on $i_{r}$ by
\vspace{-0.2cm}
\begin{equation}
a_{r}^j = \left[f(A_{r-1})-f(A_{r})\right]- \left[f(A_{r-1} \setminus \{j\}) - f(A_{r} \setminus \{j\}) \right]. 
\label{a-r-j}
\end{equation}

\end{definition}

\begin{proposition}
  For any $1 \leq r \leq l$ and $1 \leq j \leq m$ we have $a_r^j \geq 0.$
\end{proposition}
\begin{proof}

Note that $A_{r-1}= A_{r}\cup \{i_{r}\}$. Thus when  $j=i_{r}$ or $i_{r}\neq j\not \in A_{r-1}$ the second term is zero, and the result follows directly from the monotonicity of function $f$.
Assume now that $i_{r}\neq j \in A_{r-1},$ thus $j\in A_{r}$. Define $S = A_{r}$ and $T = A_{r-1} \setminus \{j\}.$ Then $S\cup T=A_{r-1}$, $S\cap T=A_{r}\setminus \{j\}$, 
and we employ the supermodularity of function $f$. \qed\end{proof}

Given an optimal solution $X=(X_{j})$, we will break it down into a large number of components $Z_{r}^{j}\in {\bf Z}$ as presented in equations~(\ref{first}) and~(\ref{second}) below:
\vspace{-0.5cm}
\begin{align}
 X_{j}= \sum_{r=1}^l Z_{r}^j, \forall j \in [m]\label{first}\\
 0 \leq Z_r^j \leq a_r^j.\label{second} 
\end{align}

Intuitively $Z_{r}^j$ is the part of the optimal solution $X_{j}$ that can be 
assigned to cover ``set $i_r$''. This explains the newly introduced constants: first, one cannot allocate more than 
the total of $X_{j}$. Second, one cannot allocate to any ``set $i_{r}$'' more than ``its intersection with $X_{j}$''.

\begin{definition} \label{Z_alpha_delta}
Given concave game $\Gamma$
\begin{itemize} 
\item  
 Let $\alpha= \alpha(\Gamma)$ the {\em smallest} positive value such that for some cover $X$ in the core minimizing $Fair_{\lambda}(\Gamma)$,  
one can define quantities $Z_{r}^{j}$ so that inequalities 
\begin{equation}
 \sum_{j=1}^{m} Z_{r}^{j} \leq \alpha \cdot \Delta_{r} 
\label{defalpha}
\end{equation}
%\vspace{-0.2cm} 
hold for any $1\leq r\leq l.$
\item
Let $\beta= \beta(\Gamma)$ the {\em largest} positive value such that for some cover $X$ in the core minimizing $Fair_{\lambda}(\Gamma)$,  
one can define quantities $Z_{r}^{j}$ so that inequalities 
\begin{equation}
 \beta \cdot \Delta_{r} \leq \sum_{j=1}^{m} Z_{r}^{j} 
\label{defbeta}
\end{equation}
%\vspace{-0.2cm} 
hold for any $1\leq r\leq l.$ 
\end{itemize} 
\end{definition}

\begin{proposition} 
For any convex game $\Gamma$ we have 
\[
\beta(\Gamma)\leq 1 \leq \alpha(\Gamma).
\]
\label{ab}
\end{proposition} 
\vspace{-0.5cm}
\begin{proof} 
We prove the first inequality, the second is proved in an entirely similar manner. 

Sum all equations~(\ref{defbeta}) for all $r=1,\ldots, l$:

 The left-hand side is 
\begin{equation*} 
 \beta(\Gamma) \cdot \sum_{r=1}^{l} \Delta_{r} = \beta(\Gamma)\cdot f(N),
\end{equation*}

by the ReverseGreedy algorithm. 

Similarly, the right-hand side is 

\begin{equation*} 
 \sum_{r=1}^{l} \left(\sum_{j=1}^{m} Z_{r}^{j}\right) = \sum_{j=1}^{m} \left(\sum_{r=1}^{l} Z_{r}^{j} \right) = \sum_{j=1}^{m} X_{j}= f(N).
\end{equation*}

The result follows. \qed
\end{proof}

Our main result gives an upper bound applicable to  all convex cooperative games: 

\begin{theorem}
Given a convex cooperative game $\Gamma$ the ReverseGreedy algorithm produces a cover $RG$ satisfying 
\begin{itemize} 
\item For 
 any $0<\lambda < 1$
\begin{equation} 
H_{\lambda}(RG)\leq H_{\lambda}(OPT)+\frac{1}{\lambda-1}\log_{2}(\beta\lambda).
\end{equation} 
\item For 
 any $\lambda> 1$
\begin{equation} 
H_{\lambda}(RG)\leq H_{\lambda}(OPT)+\frac{1}{\lambda-1}\log_{2}(\alpha\lambda).
\end{equation} 
\end{itemize} 
\label{RenyiRG}
\end{theorem} 

\begin{corollary} 
The cover RG produced by the ReverseGreedy algorithm satisfies 

\begin{equation}
D_{\lambda}(RG\parallel U) \geq  Fair_{\lambda}(\Gamma,U)-\delta_{\lambda}(\Gamma)
\end{equation}

where $\delta_{\lambda}(\Gamma)=\begin{cases}
\frac{1}{\lambda-1}\log_{2}(\beta\lambda), & \mbox{if }0<\lambda<1\\
\frac{1}{\lambda-1}\log_{2}(\alpha\lambda), & \mbox{if }\lambda>1
\end{cases}$
\label{corollaryRG}
\end{corollary}

\begin{proof} Follows directly from Theorem \ref{RenyiRG} and Lemma \ref{PQR}.\qed\end{proof}

\begin{observation} 
By Proposition~\ref{ab} both constants in the upper bounds of Theorem~\ref{RenyiRG} are nonnegative. 
\end{observation} 

\begin{observation}
 If at least one of parameters $\alpha$ or $\beta$ are equal to $1$ then we can complete the result to the case $\lambda=1$ by taking the limit $\lambda\rightarrow 1$, yielding
\begin{equation} 
H(RG)\leq H(OPT)+\log_{2}(e).
\end{equation} 

\end{observation}

\subsection{Proof of Theorem~\ref{RenyiRG}}
\begin{proof}

Let $OPT=(X_{j})_{j\in[m]}$ and $RG=(y_{i})_{i\in[n]}$.

For $1\leq r\leq l$ we will use the shorthand $U_{j}^{r}= X_{j}- \sum_{k=1}^{r} Z_{k}^{j}$.
We also define $U_{j}^{0}=X_{j}$. For any fixed $j$, sequence $(U_{j}^{r})$ is decreasing with $r$. On the other hand $U_{j}^{r-1}-U_{j}^{r}=Z_{r}^{j}$.

By the greedy choice we infer $y_{i_r}=\Delta_{r}=f(A_{r-1})-f(A_{r})$  
 for any $1\leq r\leq l,$ with $y_{i}=0$ for other values of $i.$

We have $A_{0}=N$. Thus 
\vspace{-0.2cm}
\begin{align}
& \Delta_{r}  \geq f(A_{r-1})-f(A_{r-1}\setminus \{j\}) = \notag \\
& = f(N)- [f(N)-f(A_{r-1})]-f(N\setminus\{j\}) + [f(N\setminus\{j\})-f(A_{r-1}\setminus \{j\})] \notag \\
& = f(N)-\sum_{k=1}^{r-1} [f(A_{k-1})-f(A_{k})] - f(N\setminus\{j\}) + \sum_{k=1}^{r-1} [f(A_{k-1}\setminus \{j\}) - \notag\\
& - f(A_{k}\setminus \{j\})] \geq f(N)- f(N\setminus\{j\}) - \sum_{k=1}^{r-1} a_{k}^{j} \geq X_{j}- \sum_{k=1}^{r-1} a_{k}^{j}.
\label{DeltaRAkj}
\end{align} 

At the last step we used inequality $X_{j}\leq f(N)-f(N\setminus\{j\})$, which follows from core membership (in)equalities $\sum_{k\in N\setminus \{j\}} X_{k}\geq f(N\setminus \{j\})$
 and $\sum_{k\in N} X_{k}=f(N)$.

\textbf{Case $\lambda>1$:}

First we use inequality $\sum_{j=1}^{m}Z_{r}^{j}\leq\alpha\cdot \Delta_{r}$ as follows: 
\vspace{-0.2cm}
\begin{align*}
\alpha \sum_{r=1}^{l}(\Delta_{r})^{\lambda} & =\sum_{r=1}^{l}(\alpha \Delta_{r})(\Delta_{r})^{\lambda-1}\geq\sum_{r=1}^{l}\left(\sum_{j=1}^{m}Z_{r}^{j}\right)(\Delta_{r})^{\lambda-1} \\
\end{align*}
\vspace{-0.2cm}
Applying inequality (\ref{DeltaRAkj}) we get:
\begin{align*}
 \sum_{r=1}^{l}\sum_{j=1}^{m}Z_{r}^{j}\Delta_{r}^{\lambda-1}\geq\sum_{r=1}^{l}\sum_{j=1}^{m}Z_{r}^{j}\left(X_{j}-\sum_{k=1}^{r-1}Z_{k}^{j}\right)^{\lambda-1}=
 \sum_{j=1}^{m}\left[\sum_{r=1}^{l} Z_{r}^{j} (U_{j}^{r-1})^{\lambda-1}\right]\\
\end{align*}

Using identity $Z_{r}^{j}=U_{j}^{r}-U_{j}^{r-1}$ we have: 
\begin{align*}
\sum_{j=1}^{m} \sum_{r=1}^{l} Z_{r}^{j} (U_{j}^{r-1})^{\lambda-1} = \sum_{j=1}^{m} \sum_{r=1}^{l} (U_{j}^{r}-U_{j}^{r-1}) (U_{j}^{r-1})^{\lambda-1} = (*)
\end{align*}

Transforming the difference into a sum of ones and taking into account that $x^{\lambda-1}$ is increasing we obtain:

\begin{align*}
(*) = \sum_{j=1}^{m} \left[ \sum_{r=1}^{l} \left(\sum_{k=U_{j}^{r}+1}^{U_{j}^{r-1}} 
(U_{j}^{r-1})^{\lambda-1} \right) \right] \geq \sum_{j=1}^{m} \left[ \sum_{r=1}^{l} \left( \sum_{k=U_{j}^{r}+1}^{U_{j}^{r-1}}k^{\lambda-1}\right) \right]\\
\end{align*}

From $U_j^0=X_j$ it follows that:
\begin{align*}
\sum_{j=1}^{m} \left[ \sum_{r=1}^{l} \left( \sum_{k=U_{j}^{r}+1}^{U_{j}^{r-1}}k^{\lambda-1}\right) \right] = \sum_{j=1}^{m}\left(\sum_{k=1}^{X_{j}}k^{\lambda-1}\right)\\
\end{align*}
\vspace{-0.3cm}
Putting things together, using $\lambda > 1$ and standard calculus we have: 
\begin{align*}
\alpha \sum_{r=1}^{l}(\Delta_{r})^{\lambda} \geq\sum_{j=1}^{m}\left(\sum_{k=1}^{X_{j}}k^{\lambda-1}\right) \geq \sum_{j=1}^{m}\frac{X_{j}^{\lambda}}{\lambda}=\frac{1}{\lambda}\sum_{j=1}^{m}X_{j}^{\lambda} \\
\end{align*}
\vspace{-0.4cm}
Taking the logarithm and dividing by  $1-\lambda<0$ yields:

\begin{align*}
\frac{1}{1-\lambda}\log\left(\sum_{r=1}^{l}\Delta_{r}^{\lambda}\right) \leq \frac{1}{1-\lambda}\log\left(\sum_{j=1}^{m}X_{j}^{\lambda}\right)-\frac{1}{1-\lambda}\log(\alpha\lambda)
\end{align*}

or, equivalently, by the definition of R\'{e}nyi entropy:

\begin{align*}
H_{\lambda}(RG) \leq H_{\lambda}(OPT)+\frac{1}{\lambda-1}\log(\alpha\lambda)
\end{align*}

The proof is similar in the case $0<\lambda<1$. We use instead the definition of $\beta$. Also the standard calculus inequality changes its direction. \qed
\end{proof}

\section{The worst-case fairness of Induced Subgraph Games}

In this section we particularize our general result on worst-case fairness to the case of 
 Induced Subgraph Games (problem MREIS). An easy first observation is that computing the worst-case fairness of this class of games is computationally intractable: 

\begin{theorem} For any $\lambda >0$ the following decision problems is NP-complete: 
\begin{itemize} 
\item Given an induced subgraph game $\Gamma=(G,w)$ with nonnegative weights and a constant $\eta >0$, does $Fair_{\lambda}(\Gamma,U)\geq \eta$ ? That is, is there any cover $x\in core(\Gamma)$ with $D_{\lambda}(x,U)\geq \eta$ ? 
%\item Given an induced subgraph game $\Gamma=(G,w)$ with nonnegative weights and a const%ant $\eta >0$, does $Fair_{\lambda}(\Gamma,Sh)\geq \eta$ ? That is, is there any $x\in c%ore(\Gamma)$ with $D_{\lambda}(x,Sh)\geq \eta$ ? 
\end{itemize} 
 \label{lower-meis} \end{theorem}

We believe that similar result is true for the marginalist worst-case fairness as well. However, we weren't able to prove it and leave it as an intriguing open problem (see the final section for further comments). 

Given Theorem~\ref{lower-meis} (proved in the Appendix) we need to consider 
approximation algorithms. We next give explicit upper bounds on strictly egalitarian and marginalist worst-case fairness for induced subgraph games. 
First we prove a more explicit version of Corollary \ref{corollaryRG} for such games. We then consider the target distribution $Sh$ 
obtained from the Shapley value of an induced subgraph game $\Gamma$ (see Theorem 1 in \cite{deng1994complexity}) by normalization.
In this case we compare the performance of the cover $RG$ given by the ReverseGreedy algorithm to that of a cover $BI$ corresponding to 
a biased orientation (Definition \ref{bias-orientation}).

\begin{theorem}
Given an induced subgraph game $\Gamma=(G,w)$ and $\lambda>0$
\begin{enumerate}
\item The ReverseGreedy algorithm produces a cover $RG$ satisfying 
\begin{equation} 
H_{\lambda}(Sh)-H_{\lambda}(OPT) \leq \left[H_{\lambda}(Sh) - H_{\lambda}(RG)\right]+\frac{1}{\lambda-1}\log_{2}(\lambda).
\end{equation} 
\item Consider any cover $BI$ corresponding to a biased orientation of instance $(G,w)$ of MREWO. 
It satisfies:
\begin{equation}
H_{\lambda}(Sh)-H_{\lambda}(OPT) \leq\frac{1}{\lambda}\left[H_{\lambda}(Sh) - H_{\lambda}(BI)\right]+1.
%H_{\lambda}(f_{B})\leq pSI_{\lambda}(\Gamma)+p - \bf{1}_{\{x>1\}}(p).
\end{equation}

%where we have used the notation $a\dotminus b=\max\{a-b,0\}$. 
\end{enumerate}
\label{upper-mRis}
\end{theorem}

\begin{corollary}
In the setting of the previous result we have:
\begin{enumerate}
\item For strictly egalitarian worst-case fairness:
\begin{align*}
D_{\lambda}(RG\parallel U)  \geq  Fair_{\lambda}(\Gamma,U) - \frac{1}{\lambda-1}\log_{2}(\lambda)
\end{align*}

\item For marginalist worst-case fairness:
\begin{equation}
D_{\lambda}(RG\parallel Sh)  \geq  Fair_{\lambda}(\Gamma,Sh) - \frac{1}{\lambda-1}\log_{2}(\lambda)-nu(Sh)
\label{first-guarantee}
\end{equation}
\vspace{-0.6cm}
 \begin{center}
  and
 \end{center}
\vspace{-0.1cm} 
\begin{equation}
D_{\lambda}(BI\parallel Sh)  \geq \lambda \cdot Fair_{\lambda}(\Gamma,Sh) -(1+\lambda) \cdot nu(Sh) -\lambda
\label{second-guarantee}
\end{equation}
\end{enumerate}
\label{foo}
\end{corollary}

\begin{proof} Follows directly from Theorem \ref{upper-mRis} and Lemma \ref{PQR}.\qed\end{proof}

We note that for marginalist worst-case fairness the second bound may be slightly better when $\lambda \approx 1$ %and $nu(Sh) \approx 0,$ that is when the distribution induced by the Shapley value is approximately uniform
. Indeed, in the limit $\lambda \rightarrow 1$ term $\frac{1}{\lambda-1}\log_{2}(\lambda)$ tends to $\log_{2}(e)\approx 1.442\ldots$, while $\lambda \approx 1$. This shows that the best of the two guarantees in equations~(\ref{first-guarantee}) and (\ref{second-guarantee}) may depend on the precise value of constant $\lambda$ (and, of course, other features of the instance at hand). 

\subsection{Proof of Theorem \ref{upper-mRis}.1}

The basis of our proof is the following result, that relates the problem of computing a cover in an induced subgraph game to
problem MREWO as follows:  
\begin{lemma}
\label{interred} 
Given an induced subgraph game $\Gamma=(G,w)$ the following are true: 
\begin{enumerate} 
\item The cover produced by the ReverseGreedy algorithm coincides with the solution produced by the Greedy algorithm on $\Gamma$, seen as an instance of problem MREWO. 
\item The optimal fractional solution of $\Gamma$ yields an optimal solution of instance $\Gamma$ of problem MREWO and conversely. 
\end{enumerate} 
\end{lemma}
\begin{proof}
\begin{enumerate} 
\item 
Consider a fixed step $i$ of the ReverseGreedy algorithm for problem
$\Gamma$. At this step ReverseGreedy chooses a vertex $v_{i}$ such that its
choice maximizes the marginal decrease in coalition cost. That is, the 
sum of weights of edges adjacent to $v_{i}$ and another node still in the coalition 
is the largest. 

Formally, at step $i$ we delete vertex $v_{i}\in V_{i-1}$ and obtain
the graph $G_{i}=(V_{i},E_{i})$ with $V_{i}=V_{i-1}\setminus\{v_{i}\}$
and $E_{i}=E(V_{i})$. The marginal decrease to maximize is:
\[
m(v_{i})=\sum_{v\in V_{i-1},v\neq v_{i}}w_{v_{i},v}.
\]
(note that $G_{0}=G$.) 

Similarly, given instance $(G,w)$ of problem MREWO the Greedy algorithm
(discussed at the end of Section~\ref{preliminaries}) chooses at step $i$ a vertex $v_{i}$ 
maximizing the sum of weights of all adjacent, not yet oriented, edges. 
It orients all such edges towards $v_{i}.$ This step is thus formally identical to the 
corresponding choice of ReverseGreedy on instance 
$\Gamma$ of MREIS. 

\item 

Fractional covers $Z=(z_{1},z_{2},\dots,z_{n})$ of an IS game (imputations in the core)
can be characterized as follows:

\begin{align}
z_{i} & =\sum_{(i,j)\in E}r_{i,j}w_{(i,j)}
\label{eq:xi-rij-1}
\end{align}
 where $r_{i,j}$ are real numbers in the range $0\leq r_{i,j}\leq 1$ with $r_{i,j}+r_{j,i}=1.$

This claim is an easy consequence of the characterization of
the core of IS games \cite{deng1994complexity}, and a special case of a more general paradigm \cite{deng1999algorithmic}.

\begin{lemma}
\label{claim-rij}
Let $X=(x_i)_i$ be an optimal fractional cover of instance $\Gamma.$
Without loss of generality reorder the set of vertices so that
$x_{1}\geq x_{2}\geq\cdots\geq x_{n}.$ Then coefficients $r_{i,j}$ from formula (\ref{eq:xi-rij-1}) satisfy: 

\begin{align*}
r_{i,j} & =\begin{cases}
1, &\mbox{ if } i<j\\
0, &\mbox{ if } i>j.
\end{cases}
\end{align*}

for all $i,j\in V.$\end{lemma}
\begin{proof}
Consider an arbitrary edge $(i,j)\in E.$
Assume that $i<j$ (the opposite case is easily handled via relation $r_{i,j}+r_{j,i}=1$)
and $r_{i,j}<1.$ 
Define for notational convenience 
 $0<\epsilon<1$ by $\epsilon = r_{j,i} = 1-r_{i,j}.$ With this choice further 
define $\tilde{X}=(\tilde{x_1}, \dots, \tilde{x_i}, \dots, \tilde{x_j}, \dots, \tilde{x_n})$
where 
$
\tilde{x_{i}}=x_{i}+\epsilon w_{i,j},
%\sum_{(i,k)\in E,k\neq j}r_{i,k}w_{(i,k)}+w_{(i,j)}=
$ 
$
\tilde{x_{j}}=x_{j}-\epsilon w_{i,j},
%=\sum_{(j,k)\in E,k\neq i}r_{j,k}w_{(j,k)}+0w_{(i,j)}
$
and $\tilde{x_{k}}=x_{k}$ for all other $k\neq i,j.$ 
Note that $\tilde{x}$ and $x$ differ just on components $i,j$.
By the previous remark $\tilde{x}$ is still an imputation in the core.  

We now apply Lemma~\ref{dominate_min_ent} and infer that for all 
$\lambda>0$ we have $H_{\lambda}(X)> H_{\lambda}(\tilde{X}),$
which contradicts the hypothesis that $X$ had the lowest R\'{e}nyi entropy.\qed\end{proof} 

The previous result can be interpreted as follows:
any optimal solution $X=(x_{i})_i$ satisfies $x_{i}=\sum_{(i,j)\in E,i<j}w_{(i,j)},$
where the elements of $X$ are listed in nonincreasing order.

With this result in hand, we construct a solution $Y=(y_{1},\dots,y_{n})$
of instance $\Gamma$ of MREWO by orienting all the edges $(i,j)$ with
$i<j$ towards the vertex with a lower index.  
$Y$ arises from an edge orientation $u$ given by $u((i,j))=i$ for all $i<j$.
Thus

\begin{align*}
y_{i} & =\sum_{e\in u^{-1}(i)}w_{e}=\sum_{(i,j)\in E,i<j}w_{(i,j)}=x_{i}.
\end{align*}

Clearly, the objective values of $X$ and $Y$ are the same. 

~ 

Conversely, let $Y=(y_i)_i$ be an optimal cover of an instance $\Gamma$ of
MREWO, with associated orientation $u$ and $y_{i}=\sum_{e\in u^{-1}(i)}w_{e}$. 
We assume that vertices are ordered
to respect $y_{1}\geq y_{2}\geq\cdots\geq y_{n},$ thus if $i<j$ we have $u(i,j)=i.$

From $Y$ we can define a solution $X$ of an instance $\Gamma$ of MREIS
as in equation (\ref{eq:xi-rij-1}) with
\vspace{-0.3cm}
\begin{align*}
r_{i,j} & =\begin{cases}
1, & (i,j)\in u^{-1}(i)\\
0, & \mbox{otherwise}.
\end{cases}
\end{align*}

Therefore 
\[
x_{i}=\sum_{(i,j)\in E}r_{i,j}w_{(i,j)}=\sum_{(i,j)\in u^{-1}(i)}w_{(i,j)}=y_{i},
\]
and solutions $X$ and $Y$ have again the same objective values.

Together these two facts show that optima of a common instance of the two 
problems are the same, and the proof is complete.
\qed
\end{enumerate} 
\end{proof}

To complete the proof of Theorem \ref{upper-mRis}.1, by Theorem \ref{RenyiRG} all we
need to prove is the following 

\begin{lemma} 
 For all IS games $\Gamma=(G,w)$ one can construct a system of
parameters $(Z_{r}^{j})$ from Equation (\ref{first}), witnessing equality $\alpha(\Gamma)=\beta(\Gamma)=1$. 
\end{lemma}

\begin{proof} 
The argument is a 
slight generalization of the analysis
given in \cite{istrate-bonchis2011} for the \emph{ Minimum Entropy
Orientation Problem} \cite{cardinal2008minimum,cardinal2009minimum}. 

\begin{lemma} 

Given any IS game $(G,w)$ we have 

\begin{equation}
a_{r}^{j}=\left\{ 
\begin{array}{ll}
w_{i_{r},j}, & \mbox{ if }i_{r}\neq j,(i_{r},j)\in E,j \in A_{r}\\
\Delta_{r},& \mbox{ if }i_{r}=j\\
0,&\mbox{ otherwise,}
\end{array} \right.
\end{equation}
where $\Delta_{r}$ is the value computed by the algorithm ReverseGreedy at stage $r.$
\label{coeffs}
\end{lemma} 
\begin{proof} 
A simple application of formulas~(\ref{val-is}) and (\ref{a-r-j}): in this case, for any set $S\subseteq V$ 
\[
f(S):=v(S)=\sum_{e\in S\times S} w_{e}.
\]

Therefore $f(A_{r-1})-f(A_{r})$ is the sum of weights $w_{e}$ of edges $e$ between $i_{r}$ and another node in $A_{r}$. 
On the other hand the value of expression $f(A_{r-1}\setminus\{j\})-f(A_{r}\setminus\{j\})$ depends on $j$:
\begin{itemize}
 \item It is zero if $i_{r}=j$, 
 \item It is $f(A_{r-1})-f(A_{r})$ when $j \neq i_r$ and $j\notin A_{r},$
 \item Otherwise, it is equal to the sum of weights $w_{e}$ of edges $e$ between $i_{r}$ and another node in $A_{r}\setminus\{j\}.$ 
In particular this is equal to $f(A_{r-1})-f(A_{r})$ when $j$ is not adjacent to $i_r.$
\end{itemize}
The result follows. \qed
\end{proof}

Lemma~\ref{coeffs} allows us to define a system  of coefficients $(Z_{r}^{j})$ satisfying for every $r:$
\begin{equation}
 \sum_{j}Z_{r}^{j}=\Delta_{r}.
\label{ZDelta}
\end{equation}

Together with the Definition (\ref{Z_alpha_delta}) and Proposition (\ref{ab}) this will witness the fact that $\alpha(\Gamma)=\beta(\Gamma)=1$.

In the construction we will see covers OPT and RG as edge orientations in the weighted graph $(G,w).$ 
We will use the names OPT and RG for these orientations as well. 

To enforce equation (\ref{ZDelta}) we note that $\Delta_{r}$ is the sum of weights of all edges oriented towards $i_r$ in RG. Intuitively we will
redistribute this amount among coefficients $Z_r^j$ with $1\leq j \leq m.$ We do this by comparing orientations OPT and RG.
Edges are considered in the order given by the Greedy algorithm.

\begin{itemize}
\item Start with $Z_{r}^{j}=0$ for all $r$ and $j$. 
\item Run algorithm Greedy that constructs orientation RG, updating coefficients during the algorithm:
\item {\bf At each stage $r$:} after choice of vertex $i_{r}$ we consider edges $(i_{r},j)$ oriented by Greedy towards $i_r.$ There are two possibilities:
\begin{enumerate}
\item $(i_{r},j)$ is oriented towards $i_{r}$ in both RG and OPT. We set $Z_{r}^{i_r}=Z_{r}^{i_r}+w_{i_{r},j}.$ 
\item $(i_{r},j)$ is oriented differently in OPT and RG. We let $Z_{r}^{j}=w_{i_{r},j}(=a_r^j)$ for such edges. 
\end{enumerate}
Note that the total weight assigned at stage $r$ is $\Delta_{r}$($=a_{r}^{i_{r}}$ according to Lemma \ref{coeffs}). 

Hence the inequality $0 \leq Z_r^j \leq a_r^j$ is true for $j=i_r$ as well.\qed

~

This completes the proof of Theorem \ref{upper-mRis}.1.\qed
\end{itemize}
\end{proof}

\subsection{Proof of Theorem \ref{upper-mRis}.2}

\begin{proof}
Let $\overrightarrow{G}$ be an orientation of $G=(V,\, E)$ of minimal
R\'{e}nyi entropy. Denote by $OPT=(q_i)_i$ the indegree distribution
$q_{i} =\frac{v_{\overrightarrow{G}}(i)}{W},$
 where $v_{\overrightarrow{G}}(i)$ (as in Definition \ref{bias-orientation})
is the sum of weights of all edges oriented in $\overrightarrow{G}$
towards vertex $i\in V,$ and $W$ is the sum of all edge weights. 

The R\'{e}nyi entropy of $OPT$ expands as follows:

\begin{align*}
H_{\lambda}(OPT)= & \frac{1}{1-\lambda}\log_{2}\big(\sum_{i\in V}(q_{i})^{\lambda}\big)=\frac{1}{1-\lambda}\log_{2}\left(\sum_{i\in V}\left[\frac{v_{\overrightarrow{G}}(i)}{W}\right]^{\lambda}\right)\\
= & \frac{1}{1-\lambda}\log_{2}\left(\sum_{i\in V}\frac{v_{\overrightarrow{G}}(i)}{W}\left[\frac{v_{\overrightarrow{G}}(i)}{W}\right]^{\lambda-1}\right)\\
= &\frac{1}{1-\lambda}\log_{2}\left(\frac{1}{W}\sum_{(i,j)\in\overrightarrow{G}}w_{i,j}\left[\frac{v_{\overrightarrow{G}}(i)}{W}\right]^{\lambda-1}\right)
\end{align*}

Taking into account that $x^{\lambda-1}$ is decreasing for any $0\leq\lambda<1$ we infer

\begin{align*}
H_{\lambda}(OPT)\geq & \frac{1}{1-\lambda}\log_{2}\left(\frac{1}{W}\sum_{(i,j)\in\overrightarrow{G}}w_{i,j}\left[\frac{max\{v(i),v(j)\}}{W}\right]^{\lambda-1}\right)
\end{align*}

The inequality is true for any $\lambda>1$ as well, as $x^{\lambda-1}$
is now increasing but we multiply with negative constant $\frac{1}{1-\lambda}.$

Let $G^{\flat}$ be a biased orientation. Thus

\[
v_{G^{\flat}}(i)=\sum_{\begin{array}{c}
(i,j)\in E\\
v(i)>v(j)
\end{array}}w_{i,j}.
\]
Let $BI=(q_i^{\flat})_i$ be its indegree distribution. By the definition of biasedness we have:

\begin{align*}
H_{\lambda}(OPT) \geq & \frac{1}{1-\lambda}\log_{2}\left(\frac{1}{W}\sum_{(i,j)\in\overrightarrow{G}}w_{i,j}\left[\frac{max\{v(i),v(j)\}}{W}\right]^{\lambda-1}\right)\\
= & \frac{1}{1-\lambda}\log_{2}\left(\frac{1}{W}\sum_{(i,j)\in E}w_{i,j}\left[\frac{max\{v(i),v(j)\}}{W}\right]^{\lambda-1}\right)\\
= & \frac{1}{1-\lambda}\log_{2}\left(\frac{1}{W}\sum_{(i,j)\in G^{\flat}}w_{i,j}\left[\frac{v(i)}{W}\right]^{\lambda-1}\right)\\
= & \frac{1}{1-\lambda}\log_{2}\left(\sum_{i\in V}\frac{v_{G^{\flat}}(i)}{W}\left[\frac{v(i)}{W}\right]^{\lambda-1}\right)= \\
= & \frac{1}{1-\lambda}\log_{2}\left(\sum_{i\in V}q_{i}^{\flat}\left[\frac{v(i)}{W}\right]^{\lambda-1}\right)
\end{align*}

Deng and Papadimitriou \cite{deng1994complexity} proved that the Shapley value
of an induced subgraph game is $s(i)=\frac{1}{2}\sum_{i\neq j}w_{i,j}=\frac{v(i)}{2}.$
Thus, the Shapley distribution $Sh=(s_i)_i$ of such a game is $s_{i}=\frac{v(i)}{2W}=\frac{s(i)}{W}.$ 

Hence:

\begin{align*}
H_{\lambda}(OPT) \geq & \frac{1}{1-\lambda}\log_{2}\left(\sum_{i\in V}q_{i}^{\flat}\left[\frac{2W\cdot s_{i}}{W}\right]^{\lambda-1}\right)\\
= & \frac{1}{1-\lambda}\log_{2}\left(\sum_{i\in V}q_{i}^{\flat}\left(s_{i}\right)^{\lambda-1}\right)-1
\end{align*}

The difference between entropies of the optimal and Shapley distribution can be written as follows:

\begin{align*}
& H_{\lambda}(q)-H_{\lambda}(s)\geq \frac{1}{1-\lambda}\log_{2}\left(\sum_{i\in V}q_{i}^{\flat}\left(s_{i}\right)^{\lambda-1}\right)-1-\frac{1}{1-\lambda}\log_{2}\left(\sum_{i\in V}s_{i}^{\lambda}\right)\\
& = \left[\frac{1}{1-\lambda}\log_{2}\sum_{i\in V}q_{i}^{\flat}\left(s_{i}\right)^{\lambda-1}+\frac{1}{\lambda}\log_{2}\sum_{i\in V}s_{i}^{\lambda}-\frac{1}{\lambda(1-\lambda)}\log_{2}\sum_{i\in V}\left(q_{i}^{\flat}\right)^{\lambda}\right]\\
& + \frac{1}{\lambda(1-\lambda)}\log_{2}\sum_{i\in V}\left(q_{i}^{\flat}\right)^{\lambda}-\frac{1}{\lambda}\log_{2}\sum_{i\in V}s_{i}^{\lambda}-\frac{1}{1-\lambda}\log_{2}\sum_{i\in V}s_{i}^{\lambda}-1\\
& = h_{\lambda}[BI,Sh]+\frac{1}{\lambda}H_{\lambda}(BI)-\frac{1}{\lambda}H_{\lambda}(Sh)-1
\end{align*}

Applying the discrete Gibbs Lemma we infer
\begin{align*}
H_{\lambda}(q)-H_{\lambda}(Sh) & \geq\frac{1}{\lambda}\left(H_{\lambda}(BI)-H_{\lambda}(Sh)\right)-1
\end{align*}
 which is what we had to prove, up to multiplication with -1.\qed
\end{proof}

\section*{Conclusions}

The main contribution of this paper was to propose a parametric 
family of measures of worst-case fairness in cooperative settings. Besides the analogy 
with the Cowell-Kuga measures of inequality, a partial possible justification of the parametric nature of our family arises from Corollary~\ref{foo}: depending on the particular values of $\lambda$ different approximation algorithms may provides the better approximation. 

Our analysis raises many open questions:
\begin{itemize}
\item Prove that maximizing marginalist worst-case fairness of IS games is NP-complete. We attempted to adapt the approach of \cite{cardinal2008minimum} to yield graphs $G$ that are regular (and thus the Shapley value of the IS game on $G$ coincides with the uniform distribution). However the reduction from \cite{cardinal2008minimum} seems to produce graphs that are inherently non-regular (more precisely, they have vertices with two possible degrees, 3 and 4). Further attempts at modifying the gadgets employed in this construction were unsuccessful. Nevertheless, we believe that maximizing marginalist worst-case fairness of IS games is NP-complete as well. 
\item Investigate the tightness of the upper bounds for approximating marginalistic and other worst-case fairness measures in the case of IS games. Our result showed that an additive upper bound could be obtained at least for marginalist worst-case fairness. We don't have, however, any results concerning the hardness of approximation for this measure.    
\item Extend the study of these measures to other convex cooperative games.
\end{itemize} 

\section*{Acknowledgments} 
Both authors contributed in a substantially equal manner to this work.

The first author has been supported by a project on Postdoctoral national
programs, contract CNCSIS PD\_575/2010.

The corresponding author has been supported by a project on Postdoctoral
programs for sustainable development in a knowledge-based society,
contract POSDRU/89/1.5/S/60189, co-financed by the European Social  
Fund via the Sectorial Operational Program for Human Resource Development
2007-2013.

%\bibliographystyle{alpha}
%\bibliography{/home/cosmin/Dropbox/texmf/bibtex/bib/bibtheory}
%\bibliography{/home/gistrate/Dropbox/texmf/bibtex/bib/bibtheory}

\section*{Appendix:}
\subsection*{Proof of Lemma~\ref{positiveRenyiRelative}}

We will apply the classical H\H{o}lder inequality:

\begin{align*}
\sum_{i}f_{i}g_{i} & \leq\left(\sum_{i}f_{i}^{t}\right)^{\frac{1}{t}}\left(\sum_{i}g_{i}^{r}\right)^{\frac{1}{r}}
\end{align*}

where $(f_{k})_{k},(g_{k})_{k}\in\mathbb{R}_{+},$ 
and $t,r>1$ with $1/t+1/r=1.$  There are two cases: 

Case $\lambda>1:$ Setting $t=\lambda>1$, $r=\frac{\lambda}{\lambda-1},$
$f_{i}=p_{i},$ $g_{i}=q_{i}^{\lambda-1}$ we obtain:

\begin{eqnarray*}
\sum_{i\in V}q_{i}^{\lambda-1}p_{i} & \leq & \left(\sum_{i\in V}p_{i}^{\lambda}\right)^{\frac{1}{\lambda}}\left(\sum_{i\in V}\left(q_{i}^{\lambda-1}\right)^{\frac{\lambda}{\lambda-1}}\right)^{\frac{\lambda-1}{\lambda}}
\end{eqnarray*}

Further, taking the logarithm and dividing with the negative number $1-\lambda$ we get:

\begin{align}
\frac{1}{1-\lambda}\log\left(\sum_{i\in V}q_{i}^{\lambda-1}p_{i}\right)\geq & \frac{1}{1-\lambda}\log\left[\left(\sum_{i\in V}p_{i}^{\lambda}\right)^{\frac{1}{\lambda}}\left(\sum_{i\in V}q_{i}^{\lambda}\right)^{\frac{\lambda-1}{\lambda}}\right]\nonumber \\
= & \frac{1}{\lambda(1-\lambda)}\log\sum_{i\in V}p_{i}^{\lambda}-\frac{1}{\lambda}\log\left(\sum_{i\in V}q_{i}^{\lambda}\right)\label{eq:logRenyi}
\end{align}

Case $0<\lambda<1:$ We now take $t=\frac{1}{\lambda}>1,$
$r=\frac{1}{1-\lambda}>1,$$f_{i}=\left(q_{i}^{\lambda-1}p_{i}\right)^{\lambda}$
and $g_{i}=q_{i}^{\lambda(1-\lambda)},$ therefore:

\begin{eqnarray*}
\sum_{i\in V}p_{i}^{\lambda} & = & \sum_{i\in V}\left(q_{i}^{\lambda-1}p_{i}\right)^{\lambda}q_{i}^{\lambda(1-\lambda)}\\
 & \leq & \left[\sum_{i\in V}\left(\left(q_{i}^{\lambda-1}p_{i}\right)^{\lambda}\right)^{\frac{1}{\lambda}}\right]^{\lambda}\left[\sum_{i\in V}\left(q_{i}^{\lambda(1-\lambda)}\right)^{\frac{1}{1-\lambda}}\right]^{1-\lambda}\\
 & = & \left[\sum_{i\in V}q_{i}^{\lambda-1}p_{i}\right]^{\lambda}\left[\sum_{i\in V}q_{i}^{\lambda}\right]^{1-\lambda}
\end{eqnarray*}

equivalently
\begin{eqnarray*}
\sum_{i\in V}q_{i}^{\lambda-1}p_{i} & \geq & \left(\sum_{i\in V}p_{i}^{\lambda}\right)^{\frac{1}{\lambda}}\left(\sum_{i\in V}q_{i}^{\lambda}\right)^{\frac{\lambda-1}{\lambda}}
\end{eqnarray*}

Taking the logarithm and dividing by $1-\lambda$ we obtain
the inequality. 
The result can be extended to case $\lambda = 1$ as well by taking the limit.\qed

\subsection*{Proof of Lemma~\ref{PQR}}

\begin{proof} 

From the definition we have
\begin{eqnarray*}
& & D_{\lambda}(P\parallel R)-D_{\lambda}(Q\parallel R) =\\ 
& = & \frac{1}{\lambda-1}\log_{2}\big(\sum_{i\in V}p_{i}^{\lambda}(r_{i})^{1-\lambda}\big)-\frac{1}{\lambda-1}\log_{2}\big(\sum_{i\in V}q_{i}^{\lambda}(r_{i})^{1-\lambda}\big)
\end{eqnarray*}

We will reason separately in cases $0<\lambda<1$ and $\lambda > 1.$ The result will however be the same because function 
$x^{1-\lambda}$ is increasing in the first case (decreasing in the second) while $\lambda-1$ is negative (positive).

Using inequality $r_{min}\leq r_i \leq r_{max}$ we infer 

\[
\frac{1}{\lambda-1}\log\big(\sum_{i\in V}p_{i}^{\lambda}(r_{max})^{1-\lambda}\big) \leq D_{\lambda}(P\parallel R) \leq \frac{1}{\lambda-1}\log\big(\sum_{i\in V}p_{i}^{\lambda}(r_{min})^{1-\lambda}\big)
\]
therefore:

\begin{eqnarray*}
&&D_{\lambda}(P\parallel R)-D_{\lambda}(Q\parallel R) \geq\\
& \geq & \frac{1}{\lambda-1}\log\big(\sum_{i\in V}p_{i}^{\lambda}(r_{max})^{1-\lambda}\big)-\frac{1}{\lambda-1}\log\big(\sum_{i\in V}q_{i}^{\lambda}(r_{min})^{1-\lambda}\big)\\
 & = & \frac{1}{\lambda-1}\log\big((r_{max})^{1-\lambda}\sum_{i\in V}p_{i}^{\lambda}\big)-\frac{1}{\lambda-1}\log\big((r_{min})^{1-\lambda}\sum_{i\in V}q_{i}^{\lambda}\big)\\
 & = & -\log\big(r_{max}\big)-\frac{1}{1-\lambda}\log\big(\sum_{i\in V}p_{i}^{\lambda}\big)+\log\big(r_{min}\big)+\frac{1}{1-\lambda}\log\big(\sum_{i\in V}q_{i}^{\lambda}\big)\\
 & = & H_{\lambda}(Q)-H_{\lambda}(P)-nu(R)
\end{eqnarray*}

The other inequality is proved similarly.\qed
\end{proof}

\subsection*{Proof of Theorem~\ref{lower-meis}}

The problem is easily seen to be interreducible to problem MREIS: We have $Fair_{\lambda}(\Gamma,U)\geq \eta$ if and only if there exists 
a cover $x$ of $\Gamma$, seen as an instance of MREIS such that $H_{\lambda}(x)\leq \log_{2}(n)-\eta$. Therefore it is enough to 
prove that for any $\lambda >0$ problem MREIS is NP-complete. 

Next, Lemma~\ref{interred} interreduces in effect problems MREWO and MREIS. 
For $\lambda = 1$ problem MREWO has been proved to be NP-complete \cite{cardinal2008minimum}, even in the unweighted version ($w_{e}=1$ for all $e\in E(G)$).

The proof of this result encodes an instance $\Phi$ of an NP-complete version of $1$-in-$k$ SAT into a graph $G=(V,E)$ and an orientation $\overrightarrow{G}$ of 
whose in-degree distribution dominates all other probability distributions 
arising from edge orientations of $G.$ We may read the 
existence of a satisfying assignment for $\Phi$ from the entropy of orientation 
$\overrightarrow{G}$. 

The proof can be easily adapted to the general case $\lambda >0$ by using Lemma \ref{dominate_min_ent} which is the generalization of the result 
employed in \cite{cardinal2008minimum} to all values $\lambda > 0$. The basis is the following modification of Claim 2 in \cite{cardinal2008minimum},
 with essentially the same proof as the original version. 

\begin{lemma}
For any $\lambda > 0$, $\lambda \neq 1$, the R\'enyi entropy of the distribution corresponding to orientation $\overrightarrow{G}$ 
(constructed in \cite{cardinal2008minimum}) is at most 
\begin{eqnarray*}
\log_{2}m-\frac{1}{1-\lambda}\log_{2}\left(4^{\lambda}+7\cdot3^{\lambda-1}+1\right)/12
\end{eqnarray*}

with equality if and only if instance $\Phi$ is satisfiable. 
\end{lemma}

The conclusion of this argument is that computing lower bounds on 
$Fair_{\lambda}(\Gamma,U)$ is NP-complete (the fact that they the problem is in NP is trivial). \qed

\end{document}